\documentclass[11pt,reqno]{amsart}

\usepackage[english]{babel}
\usepackage{amsmath,amssymb,amsthm}
\usepackage{amsfonts}
\usepackage{amsxtra}
\usepackage{array,dcolumn}
\usepackage{graphicx}
\usepackage{hyperref}   
\hypersetup{
 pdfborder={0 0 0},      
 colorlinks=true,        
 linktoc=page,           
 linkcolor=blue,         
 citecolor=blue,         
}

\usepackage[T1]{fontenc}
\usepackage[utf8]{inputenc}
\usepackage{bbm}
\usepackage[sc]{mathpazo}

 \calclayout

\newcommand{\mathsym}[1]{{}}
\newcommand{\unicode}[1]{{}}

\theoremstyle{plain}
\newtheorem{theorem}{Theorem}
\newtheorem{lemma}[theorem]{Lemma}
\newtheorem{corollary}[theorem]{Corollary}
\newtheorem{proposition}[theorem]{Proposition}
\numberwithin{theorem}{section}

\theoremstyle{definition}

\newtheorem{conjecture}[theorem]{Conjecture}

\theoremstyle{remark}
\newtheorem{remark}[theorem]{Remark}

\numberwithin{equation}{section}

\newcommand{\Z}{\mathbb Z}

\newcommand{\half}{\tfrac{1}{2}}

\newcommand{\E}{\mathbb{E}}

\newcommand{\one}{\mathbbm{1}} 

\DeclareMathOperator{\tr}{Tr} 

\newcommand{\MeijerG}[8][\Big]{G^{{ #2 },{ #3 }}_{{ #4 },{ #5 }} #1( \begin{matrix} #6 \\ #7 \end{matrix}\, #1\vert\, #8 #1)}
\newcommand{\hypergeometric}[6][\bigg]{\,{}_{#2} F_{#3} #1( \begin{matrix} #4 \\ #5 \end{matrix}\, #1\vert\, #6 #1)}

\begin{document}


\title[How many eigenvalues of a product of truncated orthogonal matrices are real?]%
{How many eigenvalues of a product\\ of truncated orthogonal matrices are real?}
\author{P. J. Forrester}
\address{ARC Centre of Excellence for Mathematical and Statistical Frontiers, School of Mathematics and Statistics, The University of Melbourne, Victoria 3010, Australia}
\email{pjforr@unimelb.edu.au}
\author{J. R. Ipsen}
\address{ARC Centre of Excellence for Mathematical and Statistical Frontiers, School of Mathematics and Statistics, The University of Melbourne, Victoria 3010, Australia}
\email{jesper.ipsen@unimelb.edu.au}
\author{S. Kumar}
\address{Department of Physics, Shiv Nadar University, Gautam Buddha Nagar, Uttar Pradesh - 201314, India}
\email{skumar.physics@gmail.com}
\date{\today}


\begin{abstract}
A truncation of a Haar distributed orthogonal random matrix gives rise to a matrix whose eigenvalues are either real or complex conjugate pairs, and are supported within the closed unit disk. This is also true for a product $P_m$ of $m$ independent truncated orthogonal random matrices. One of most basic questions for such asymmetric matrices is to ask for the number of real eigenvalues. In this paper, we will exploit the fact that the eigenvalues of $P_m$ form a Pfaffian point process to obtain an explicit determinant expression for the probability of finding any given number of real eigenvalues. We will see that if the truncation removes an even number of rows and columns from the original Haar distributed orthogonal matrix, then these probabilities will be rational numbers. Finally, based on exact finite formulae, we will provide conjectural expressions for the asymptotic form of the spectral density and the average number of real eigenvalues as the matrix dimension tends to infinity.

\end{abstract}


\maketitle


\section{Introduction}\label{s1}

In the study of random real symmetric matrices, the notion of an orthogonally invariant probability density function (PDF) is of primary importance. Let $X$ be an $N$-by-$N$ symmetric random matrix and let the PDF (with respect to the flat measure) be denoted $P(X)$. Orthogonal invariance means that
\begin{equation}\label{1.1}
P(Q^TXQ)=P(X)
\end{equation}
for all real orthogonal matrices $Q\in O(N)$. Since real symmetric matrices are diagonalised by real orthogonal matrices, a corollary is that $P$ depends only on the eigenvalues. The latter feature is to be combined with the fact that the volume element $(\mathrm{d}X)=\prod_{1\leq i\leq j\leq N}\,\mathrm{d}X_{ij}$, when written in terms of the eigenvalues $\{\lambda_i\}$ and eigenvectors $\{q_i\}$, factorises according to
\begin{equation}\label{1.2}
(\mathrm{d}X)=\prod_{1\leq i<j\leq N}|\lambda_j-\lambda_i|(Q^T\mathrm{d}Q),
\end{equation}
where $(Q^T\mathrm{d}Q)$ is the invariant measure for the matrix of eigenvectors $Q=[q_1,\ldots,q_N]$, see e.g. \cite[Eq. (1.11)]{Fo10}. One then has for the eigenvalue PDF the functional form
\begin{equation}\label{1.3}
C_NP\left(\mathrm{diag}(\lambda_1,\ldots,\lambda_N)\right)\prod_{1\leq i<j\leq N}|\lambda_j-\lambda_i|,
\end{equation}
where $C_N$ is a normalisation constant given by integration over the eigenvectors.

Now, suppose that when restricted to diagonal matrices, $P(X)$ exhibits the further structure
\begin{equation}\label{1.4}
P\left(\mathrm{diag}(\lambda_1,\ldots,\lambda_N)\right)=\prod_{l=1}^Nw(\lambda_l).
\end{equation}
Important examples in random matrix theory include the classical Hermite, Laguerre, Jacobi, and Cauchy matrix weights given by
\begin{equation*}
e^{-\mathrm{Tr}X^2},\qquad 
\mathrm{det}X^{\alpha}e^{-\mathrm{Tr}X}\,\one_{X>0},\qquad
\det X^{\alpha}\det(\mathbbm{1}-X)^{\beta}\,\one_{0<X<1}, \qquad\text{and}\qquad 
\det(\mathbbm{1}+X^2)^{-\gamma},
\end{equation*}
respectively.
Here $\one_J$ is the indicator function (i.e. $\one_J=1$ if $J$ is true and $\one_J=0$ otherwise), and the matrix inequality $A>B$ for symmetric matrices $A$ and $B$ should be read as: `$A-B$ is positive definite'. Another example satisfying
(\ref{1.4}) is the family of PDFs
\begin{equation*}
e^{-\tr V(X)}\qquad \text{with} \qquad V(X)={\sum_{l=1}^{\infty}t_lX^l}
\end{equation*}
indexed by the infinite sequence $\{t_l\}_{l=1}^{\infty}$, constrained only by suitable decay at infinity. We remark that it is fundamental to random matrix theory that if the separation property~\eqref{1.4} holds, then the eigenvalue PDF~\eqref{1.3} corresponds to a Pfaffian point process (see e.g. \cite[Ch. 6]{Fo10} and Section \ref{s4.1} below).

Rather than symmetric matrices, consider instead an $N$-by-$N$ asymmetric random real matrix, $X$. Now, real orthogonal matrices $Q$ can no longer be used to transform $X$ into diagonal matrix form. However, a transformation to a block upper triangular form can still be obtained according to the real Schur decomposition
\begin{equation}\label{1.5}
X=Q(D^{(k)}+T)Q^T.
\end{equation}
Here the superscript $k$ labels the number of real eigenvalues ($k$ must then have the same parity as $N$, i.e. $k\equiv N\mod2$); the remaining $N-k$ eigenvalues appear as complex conjugate pairs. The matrix $D^{(k)}$ is block diagonal with the first $k$ diagonal entries the real eigenvalues of $X$, $\{\lambda_1,\ldots,\lambda_k\}$, and the next $(N-k)/2$ block entries the $2\times2$ real matrices $\{G_s\}_{s}$ with complex eigenvalues $\{x_s\pm\mathrm{i}y_s\}_{s}$ coinciding with the complex eigenvalues of $X$. The matrix $T$ is a strictly upper triangular matrix.

Analogous to \eqref{1.2}, in terms of these variables the volume element $(\mathrm{d}X)=\prod_{i,j=1}^N\,\mathrm{d}X_{ij}$ transforms according to
\begin{equation}\label{1.6}
(\mathrm{d}X)=\prod_{j<p}|\lambda(D^{(k)}_{pp})-\lambda(D^{(k)}_{jj})|\,(\mathrm{d}T)(Q^T\mathrm{d}Q)\prod_{j=1}^k\mathrm{d}\lambda_j\prod_{s=1}^{(N-k)/2}\mathrm{d}G_s,
\end{equation}
where $\lambda(D^{(k)}_{pp})$ refers to the eigenvalues of the $(pp)$\textsuperscript{th} block entry of $D^{(k)}$. In particular, the measure again factorises. Substituting \eqref{1.5} into \eqref{1.1} shows
\begin{equation*}
P(QXQ^T)=P(D^{(k)}+T),
\end{equation*}
so in the case that $P$ is orthogonally invariant, the dependence on $Q$ contributes only to the normalisation of the eigenvalue PDF just as for symmetric matrices. On the other hand, in distinction to the circumstance for real symmetric matrices, the calculation of the eigenvalue PDF still requires that $P$ be integrated over the triangular matrix $T$, giving in place of \eqref{1.3} the expression
\begin{multline}\label{1.7}
{C_N}\sum_k\frac{1}{k!\left((N-k)/2\right)!}\int\,P(D^{(k)}+T)(\mathrm{d}T)\prod_{s=1}^{(N-k)/2}\delta\left(G_s-\mathrm{diag}(x_s\pm\mathrm{i}y_s)\right)(\mathrm{d}G_s)
\\ \times\prod_{j<p}|\lambda(D_{pp})-\lambda(D_{jj})|.
\end{multline}
Here, $C_N$ is a constant coming from integration over $Q$ and the binomial type factor arises from relaxing the ordering needed for \eqref{1.5} to be one-to-one (the sum over $k$ includes only terms with the same parity as $N$).
It has been known since the work of Sinclair \cite{Si06} that in the circumstance that
\begin{equation}\label{1.8}
\int\,P(D^{(k)}+T)(\mathrm{d}T)\int\prod_{s=1}^{(N-k)/2}\delta\left(G_s-\mathrm{diag}(x_s\pm\mathrm{i}y_s)\right)(\mathrm{d}G_s)
=\prod_{l=1}^kw_r(\lambda_l)\prod_{j=1}^{(N-k)/2}w_c(x_j,y_j)
\end{equation}
for some weights $w_r(\lambda)$ and $w_c(x,y)$, then \eqref{1.7} corresponds to a two-component --- the real and complex eigenvalues --- Pfaffian point process. However, the choices of $P$ which give rise to \eqref{1.8} are far more restrictive than those for symmetric matrices permitting the factorisation \eqref{1.4}. 

The first identified case of~\eqref{1.8} was that of standard real Gaussian matrices, corresponding to $P(X)$ proportional to $e^{-\mathrm{Tr}X^TX}$ \cite{LS91, Ed97}. Some years later, the product matrices $X=X_1^{-1}X_2$ and $X=X_1X_2$ with each $X_i$ a standard real Gaussian matrix, were shown to be further examples \cite{APS10, FM11}, as  was $X$ defined as an $N\times N$ sub-block of a $(N+L)\times(N+L)$ real orthogonal matrix \cite{KSZ09}. Ipsen and Kieburg \cite{IK14} extended these results to an arbitrary sized matrix product
\begin{equation}\label{1.9}
P_m=X_1X_2\cdots X_m
\end{equation}
with each $X_i$ either a standard real Gaussian matrix or a truncation of a real orthogonal matrix. For a product of real Gaussian matrices, probabilistic and statistical quantities of the Pfaffian point process formed by the eigenvalues were calculated and analysed in the recent work \cite{FI16}; see also \cite{Si16}. It is our purpose in the present work to undertake an analogous study of the Pfaffian point process for the eigenvalues of the matrix product \eqref{1.9} with each $X_i$ the truncation of a real orthogonal matrix. A first step in this direction has been made in another recent work \cite{FK16}, in which determinantal formulae were given for the probability that all eigenvalues are real, and their arithmetic properties were analysed. It was also seen the probability that all eigenvalues are real tends to unity when the number of factors tends to infinity; this is part of much general result expected to hold for products of random matrices~\cite{La13,Fo14,Ip15lya,HJL15,AI15,Ip15,Re16,Re17}.

To undertake this study requires first revisiting the work of \cite{IK14} on the eigenvalue PDF for products of truncations of real orthogonal matrices. It turns out that the form given therein does not explicitly isolate the functional forms $w_r$ and $w_c$ in \eqref{1.8}. Rather it treats the real and complex eigenvalues on an equal footing, which is not optimal for our purposes. In Section~\ref{s2} we provide the functional form for the eigenvalue PDF of a product of $m$ matrices given as $N\times N$ sub-blocks of $(N+L)\times(N+L)$ real orthogonal Haar distributed random matrices. This generalises the $m=1$ result found by Khoruzheko, Sommers and Zyczkowski~\cite{KSZ09}. Under the constraint that there are exactly $k$ real eigenvalues ($k$ of the same parity as $N$), this PDF with $\lambda_l \in (-1,1)$
and $(x_j,y_j) \in D_+$, where $D_+$ denotes the open half unit disk $|z|<1 $ and $y>0$, is equal to
\begin{equation}
 \label{1.10}
\frac{K_{N,L}}{k!\left((N-k)/2\right)!}\left|\Delta\left(\{\lambda_l\}_{l=1}^k\cup\{x_j\pm\mathrm{i}y_j\}_{j=1}^{(N-k)/2}\right)\right|
\prod_{j=1}^kw(\lambda_j;L)\prod_{j=1}^{(N-k)/2}2\left(w\left((x_j,y_j);L\right)\right)^2,
\end{equation}
where 
\begin{equation}
\Delta\left(\{z_l\}_{k=1}^p\right)=\prod_{1\leq i<j\leq p}(z_i-z_j)
\end{equation}
denotes the Vandermonde determinant. With
\begin{equation}\label{1.11}
\mathrm{vol}\left(O(p)\right)=\frac{2^p\pi^{p(p+1)/4}}{\prod_{j=1}^p\Gamma(j/2)}
\end{equation}
being the volume of the orthogonal group, we have (\cite[Below eq. (6)]{KSZ09} contains a typo, which was corrected in \cite{Ma11})
\begin{equation}\label{1.12}
K_{N,L}
=\frac{\mathrm{vol}\left(O(L)\right)\mathrm{vol}\left(O(N)\right)}{\mathrm{vol}\left(O(L+N)\right)}\left(\frac{(2\pi)^L}{L!}\right)^{N/2}
=\bigg(\frac{2^L}{L!}\bigg)^{\!N/2}\prod_{j=1}^N\frac{\Gamma(\frac{L+j}2)}{\Gamma(\frac j2)}.
\end{equation}
Furthermore, the weight function is given by
\begin{equation}\label{1.13}
w(z;L)=\left\{\begin{array}{ll}\left(\frac{L(L-1)}{2\pi}|1-z^2|^{L-2}\int_{2|\mathrm{Im}\,z|/|1-z^2|}^1(1-t^2)^{(L-3)/2}\mathrm{d}t\right)^{1/2},&\quad L>1
\\\left(\frac{1}{2\pi}\right)^{1/2}|1-z^2|^{-1/2},&\quad L=1\end{array}\right. .
\end{equation}

Our main result, stated in Theorem \ref{theorem1}, identifies both $w_r$ and $w_c$ in \eqref{1.7}. However, as already present in the study of the eigenvalues of the product \eqref{1.8} for each $X_i$ a real standard Gaussian \cite{FI16}, the expression for $w_c$ is too complicated for further analysis (unless $m=1$) so we restrict attention to the computation of statistical and probabilistic properties of the real eigenvalues. In Section \ref{s3} we give a determinantal formula (with entries given by certain Meijer $G$-functions) for the probabilities $p_{N,k}^{P_m}$ that the product matrix~\eqref{1.9} has exactly $k$ real eigenvalues. In the case $k=N$, recent work \cite{FK16} has demonstrated special arithmetic properties of these probabilities. We further consider this theme, as well as some questions relating to the large $N$ asymptotics. In Section \ref{s4} the explicit form of the $k$-point correlation function for the real eigenvalues is presented, with the case $k=1$ corresponding to the density of the real eigenvalues. This allows various scaling limits to be analysed, and a formula for the expected number of real eigenvalues to be presented.

\section{The eigenvalue PDF}\label{s2}

Consider an $(L_i+N)\times(L_i+N)$ real orthogonal matrix chosen with Haar measure. Denoting by $X_i$ an $N\times N$ sub-block, it is straightforward to show (see e.g. \cite[\S3.8.2]{Fo10}) that $X_i^TX_i$ will have $N-L_i$ eigenvalues equal to unity for $N>L_i$. This implies that the distribution of $X_i$ is then singular. On the other hand, for $N\leq L_i$ the distribution is absolutely continuous with density \cite{Fo06a, KSZ09}
\begin{equation}\label{2.1}
P(X)=C_{N,L_i}\det(\mathbbm{1}-X_i^TX_i)^{(L_i-N-1)/2}.
\end{equation}
with
\begin{equation}\label{constant-trunc}
C_{N,L_i}=\frac{\left(\mathrm{vol}\,O(L_i)\right)^2}{\mathrm{vol}\,O(L_i+N)\mathrm{vol}\,O(L_i-N)}
=\frac{1}{\pi^{N^2/2}}\prod_{j=1}^N\frac{\Gamma(\frac{L_i+j}2)}{\Gamma(\frac{L_i-N+j}2)}
\end{equation}
Not surprisingly, the calculation of the eigenvalue PDF of $X_i$ is much simpler in this setting (compare the derivation of \eqref{1.10} given in \cite{Ma11} to that given in \cite{Fi12}, with the latter providing the derivation of \eqref{1.10} that was sketched in \cite{KSZ09}). Since the final functional form is insensitive to this detail, we will proceed with our derivation of the eigenvalue PDF for the product \eqref{1.9} assuming that each $X_i$ has a density \eqref{2.1}. 

\begin{theorem}\label{theorem1}
Consider the matrix product \eqref{1.9} with each $X_i$ the $N\times N$ sub-block of an $(L_i+N)\times(L_i+N)$ real orthogonal matrix. Given that there are $k$ real eigenvalues, where $k$ is of the same parity as $N$, the eigenvalue PDF,
supported on the same domain as (\ref{1.10}), is
\begin{equation}\label{2.2}
\frac{\prod_{i=1}^mK_{N,L_i}}{k!\left((N-k)/2\right)!}\left|\Delta\left(\{\lambda_l\}_{l=1}^k\cup\{x_j\pm\mathrm{i}y_j\}_{j=1}^{(N-k)/2}\right)\right|\prod_{j=1}^kw_r^{(m)}(\lambda_j)\prod_{j=1}^{(N-k)/2}w_c^{(m)}((x_j,y_j))
\end{equation}
where $K_{N,L_i}$ is given by \eqref{1.12},
\begin{equation}\label{2.3}
w_r^{(m)}(\lambda)=\int_{(0,1)^m}
\mathrm{d}\lambda^{(1)}\cdots\mathrm{d}\lambda^{(m)}\delta(\lambda-\lambda^{(1)}\cdots\lambda^{(m)})\prod_{l=1}^mw(\lambda^{(l)};L_l)
\end{equation}
with
\begin{equation}\label{2.4}
w(\lambda;L)=\frac{(1-\lambda^2)^{L/2-1}}{\sqrt{2\pi}}
\left(L\frac{\Gamma(1/2)\Gamma\left((L+1)/2\right)}{\Gamma(L/2)}\right)^{1/2},
\end{equation}
and
\begin{equation}\label{2.5}
w_c^{(m)}((x,y))=\int_0^{1-x^2-y^2}\mathrm{d}\delta\frac{\delta}{\sqrt{\delta^2+4y^2}}W\left(\begin{bmatrix}\mu_+&0\\0&\mu_-\end{bmatrix}\right)
\end{equation}
with
\begin{equation}\label{2.6}
\mu_{\pm}=\half\left(\pm|\delta|+(\delta^2+4(x^2+y^2))^2\right)
\end{equation}
and
\begin{multline}\label{2.7}
W(G)=\prod_{i=1}^m\frac{L_i(L_i-1)}{\pi}\int_{|G^{(i)}| < 1}(\mathrm{d}G^{(i)})\mathrm{det}\left(\mathbb{I}_2-G^{(i)}{G^{(i)}}^T\right)^{(L_i-3)/2}\delta(G-G^{(1)}\cdots G^{(m)}).
\end{multline}
\end{theorem}
\begin{proof}
According to \eqref{2.1}, and with $C_{N,L_i}$ given by~\eqref{constant-trunc} (under the assumption that $N\leq L_i$) the joint probability measure for $\{X, X_1, \ldots, X_m\}$ is
\begin{equation}\label{2.8}
\delta(X-X_1\cdots X_m)\prod_{i=1}^mC_{N,L_i}\,\mathrm{det}\left(\mathbb{I}-X_i^TX_i\right)^{(L_i-N-1)/2}(\mathrm{d}X_i)(\mathrm{d}X).
\end{equation}
For the matrices $X_i$ we follow the strategy used in \cite{FK16, FI16}
and use a generalised real Schur decomposition
\begin{equation}\label{2.9}
X_i=Q_i(D_i^{(k)}+T_i)Q_{i+1}^{-1},\quad(i=1,\ldots,m),
\end{equation}
where $Q_{m+1}:=Q_1$. With $O^*(N)$ deformed to be the set of matrices in $O(N)$ with the first entry in each column positive, each $Q_i$ in \eqref{2.9} is a real orthogonal matrix in $O^*(N)/O^*(2)^{(N-k)/2}$. Each matrix $D_i^{(k)}$ is a block diagonal matrix with the first $k$ diagonal entries scalars $\{\lambda_1^{(i)},\ldots,\lambda_k^{(i)}\}$ and the next $(N-k)/2$ block diagonal entries $2\times2$ matrices $\{G_s^{(i)}\}_{s=1}^{(N-k)/2}$. The matrices $T_i$ are each strictly upper triangular.

Introduce the block diagonal product $D=D_1\cdots D_m$ and denote the first $k$ diagonal entries $\{\lambda_t:=\lambda_t^{(1)}\cdots\lambda_t^{(m)}\}_{t=1}^k$, and the latter $(N-k)/2$ block $2\times 2$ matrices $\{G_s:=G_s^{(1)}\cdots G_s^{(m)}\}_{s=1}^{(N-k)/2}$. We know that the Jacobian for the change of variables is then \cite[Prop. A.26]{Ip15}
\begin{multline}\label{2.10}
\prod_{l=1}^m(\mathrm{d}X_l)=\prod_{j<p}|\lambda(D_{pp})-\lambda(D_{jj})|\prod_{l=1}^m(\mathrm{d}T_l)(Q_l^T\mathrm{d}Q_l)\prod_{l=1}^m\left(\prod_{j=1}^k\mathrm{d}\lambda_j^{(l)}\prod_{s=1}^{(N-k)/2}\mathrm{d}G_s^{(l)}\right),
\end{multline}
which generalises \eqref{1.6}. 
It follows that the PDF for $\{\lambda_t\}_{t=1}^k\cup\{G_s\}_{s=1}^{(N-k)/2}$ is equal to
\begin{multline}\label{2.12}
\frac{\prod_{i=1}^mC_{N,L_i}}{k!\left((N-k)/2\right)!}\prod_{j<p}|\lambda(D_{pp})-\lambda(D_{jj})|\prod_{j=1}^k\int\mathrm{d}\lambda_j^{(1)}\cdots\mathrm{d}\lambda_j^{(m)}\delta(\lambda_j-\lambda_j^{(1)}\cdots\lambda_j^{(m)})
\\ \times\prod_{s=1}^{(N-k)/2}\int(\mathrm{d}G_s^{(1)})\cdots(\mathrm{d}G_s^{(m)})\delta(G_s-G_s^{(1)}\cdots G_s^{(m)})
\\ \times\prod_{i=1}^m\int(\mathrm{d}T_i)(Q_i^T\mathrm{d}Q_i)\det(\mathbb{I}-X_i^TX_i)^{(L_i-N-1)/2},
\end{multline}
where as in \eqref{1.7} and \eqref{1.10}, the combinatorial prefactor results from relaxing the ordering on the eigenvalues required to make \eqref{1.5} one-to-one.

From the definitions, we see that
\begin{equation}\label{2.13}
\prod_{j<p}|\lambda(D_{pp})-\lambda(D_{jj})|=\left|\Delta\left(\{\lambda_l\}_{l=1}^k\cup\{x_j\pm\mathrm{i}y_j\}_{j=1}^{(N-k)/2}\right)\right|\prod_{j=1}^{(N-k)/2}\frac{1}{2y_j}
\end{equation}
and we know too \cite{Ed97} that an orthogonal similarity transformation can be used to bring each $G_{\mu}$ into the form
\begin{equation*}
\begin{bmatrix}x_{\mu}&b_{\mu}\\-c_{\mu}&x_{\mu}\end{bmatrix},
\end{equation*}
with $b_{\mu},c_{\mu}>0$, showing that the eigenvalues are $x_{\mu}\pm\mathrm{i}y_{\mu}$ with $y_{\mu}^2=b_{\mu}c_{\mu}$. From this latter point, we may change variables from the elements of $G_{\mu}$ to $\{x_{\mu},y_{\mu},\delta_{\mu},\theta_{\mu}\}$, where $\delta_{\mu}=b_{\mu}-c_{\mu}$ and $\theta_{\mu}$ parametrises the orthogonal similarity transformation. Integrating out the latter, the Jacobian of the transformation is
\begin{equation*}
\frac{4\pi y_{\mu}|\delta_{\mu}|}{\sqrt{\delta_{\mu}^2+4y_{\mu}^2}}
\end{equation*}
(see e.g. \cite[Proof of Prop. 15.10.1 and Prop. 15.10.2]{Fo10}). Also, due to the left and right orthogonal invariance each matrix $G_s$ in \eqref{2.12} may be replaced by its singular values as given in \eqref{2.6}.

Taking into consideration the theory of the above paragraph, and noting in particular the structure of \eqref{2.12}, we see that \eqref{2.2} is true for general $m\geq1$ provided it is true for $m=1$. For $m=1$, comparison of \eqref{2.2} and \eqref{2.5} with \eqref{1.10} and \eqref{1.13} shows that the task is to verify that
\begin{equation}\label{2.14}
\int_0^{1-x^2-y^2}\frac{\delta((1-\mu_{+}^2)(1-\mu_{-}^2))^{(L-3)/2}}{\sqrt{\delta^2+4y^2}}\mathrm{d}\delta
=|1-z^2|^{L-2}\int_{2|\mathrm{Im}\,z|/|1-z^2|}^1(1-t^2)^{(L-3)/2}\mathrm{d}t,
\end{equation}
where $\mu_{\pm}$ is given by \eqref{2.6}. From the latter we can check
\begin{equation*}
(1-\mu_{+}^2)(1-\mu_{-}^2)=(1-x^2-y^2)-\delta^2.
\end{equation*}
Also, changing variables $s=(\delta^2+4y^2)^{1/2}$ the integral on the LHS reads
\begin{equation*}
\int_{2y}^{\left((1-x^2-y^2)^2+4y^2\right)^{1/2}}\left((1-x^2-y^2)^2+4y^2-s^2\right)^{(L-3)/2}\mathrm{d}s.
\end{equation*}
Setting $s=\left((1-x^2-y^2)^2+4y^2\right)^{1/2}t=|1-z^2|t$, this is seen to equal the RHS.
\end{proof}

\begin{remark}
The equation (\ref{2.8}) is not valid for parameters $L_i < N$ since the density function for $X_i$ is then singular.
To proceed, following \cite{KSZ09,Fi12}, let $Y_i$ be the $L_i \times N$ rectangular matrix, which when appended to
the bottom of $X_i$ gives the first $N$ columns of the $(L_i+N) \times (L_i+N)$ real orthogonal matrix. One then
has that the joint distribution of $\{X_i, Y_i\}$ is given by the distribution
\begin{equation}\label{sXY}
\tilde{c} \delta (X_i^T X_i + Y_i^T Y_i  - I_N)
\end{equation}
where the normalisation $\tilde{c}$ is given by \cite[Eq.~(2.0.15)]{Fi12}. The calculation is thus more complicated
due to the involvement of the auxiliary variables implied by $Y_i$. For the case $m=1$, all the required working
is given in \cite[\S 4.2.2]{Fi12}. But as in our proof above for the cases $L_i \ge N$, the structure of the
analogue of (\ref{2.12}), obtained by replacing $\det(\mathbb I - X_i^T X_i)^{(L_i - N - 1)/2}$ therein by
(\ref{sXY}), and further integrating over $Y_i$, we see that again (\ref{2.3}) is true for $m \ge 1$
conditional only on it being true for $m=1$.
\end{remark}

\section{Probability of \textit{k} real eigenvalues}\label{s3}

\subsection{The generating function as a Pfaffian}\label{s3.1}

Use $Q_{N,k}(\{\lambda_l\}_{l=1}^k,\{x_j\pm\mathrm{i}y_j\}_{j=1}^{(N-k)/2})$
to denote the PDF \eqref{2.2}.
The probability $p_{N,k}^{P_m}$ of there being precisely $k$ real eigenvalues ($k$ the same parity as $N$) is then
\begin{equation}\label{3.1}
p_{N,k}^{P_m}=\prod_{l=1}^k\int_{-1}^{+1}\mathrm{d}\lambda_l\prod_{j=1}^{(N-k)/2}\int_{D_{+}}\mathrm{d}x_j\mathrm{d}y_j Q_{N,k}\left(\{\lambda_l\}_{l=1}^k,\{x_j\pm\mathrm{i}y_j\}_{j=1}^{(N-k)/2}\right),
\end{equation}
where $D_+=\{(x,y)\,:\,x^2+y^2<1 \text{ and } y>0\}$ denotes the half unit disk.
A fundamental feature of \eqref{3.1}, which follows from the structure of \eqref{2.2}, is that the corresponding generating function
\begin{equation}\label{3.2}
Z_N(\zeta)=\sum_{\substack{k=0\\ k\equiv N\mod 2}}^N\zeta^kp_{N,k}^{P_m},
\end{equation}
can be written as a Pfaffian. This was observed by Sinclair in the case $\zeta=1$, and the details of the necessary working can be found in \cite[Prop. 15.10.3, $N$ even]{Fo10} and \cite[\S4.3.1 $N$ even and \S4.3.2 $N$ odd]{Ma11}. The final result is reported in \cite[Prop. 5]{FI16}. We repeat it here, allowing for minor changes in notation.

\begin{proposition}\label{prop1}
Let $\{p_{l-1}(x)\}_{l=1,\ldots,N}$ be a set of monic polynomials, with $p_{l-1}(x)$ of degree $l-1$. Let
\begin{align}
\alpha_{j,k}&=\int_{-1}^{1}\mathrm{d}x\,w_r^{(m)}(x)\int_{-1}^{1}\mathrm{d}y\,w_r^{(m)}(y)p_{j-1}(x)p_{k-1}(y)\mathrm{sgn}(y-x),\nonumber
\\ \beta_{j,k}&=2\mathrm{i}\int_{D_{+}}\mathrm{d}x\mathrm{d}y\,w_c^{(m)}(x,y)\left(p_{j-1}(x+\mathrm{i}y)p_{k-1}(x-\mathrm{i}y)-p_{k-1}(x+\mathrm{i}y)p_{j-1}(x-\mathrm{i}y)\right),\label{3.3}
\end{align}
and
\begin{equation}\label{3.4}
\mu_k=\int_{-1}^{1}w_r^{(m)}(x)p_{k-1}(x)\mathrm{d}x.
\end{equation}
For $N$ even
\begin{equation}\label{3.5}
Z_N(\zeta)=\left(\prod_{i=1}^mK_{N,L_i}\right)\mathrm{Pf}\left[\zeta^2\alpha_{j,l}+\beta_{j,l}\right]_{j,l=1,\ldots,N}
\end{equation}
while for $N$ odd
\begin{equation}\label{3.6}
Z_N(\zeta)=\zeta\left(\prod_{i=1}^mK_{N,L_i}\right)\mathrm{Pf}\begin{bmatrix}{[}\zeta^2\alpha_{j,l}+\beta_{j,l}]&[\mu_j]\\{[}-\mu_l]&0\end{bmatrix}_{j,l=1,\ldots,N}.
\end{equation}
\end{proposition}

\subsection{Skew orthogonal polynomials}\label{s3.2}
With $w_r^{(m)}$ given by \eqref{2.3}, the double integral defining $\alpha_{j,k}$ can be computed in terms of a particular Meijer G-function \cite{FK16}; see \eqref{3.14} below. On the other hand, a direct computation of $\beta_{j,k}$ does not appear possible; recall the definition \eqref{2.5} of $w_c^{(m)}$ therein. Fortunately, an indirect evaluation is possible, provided the monic polynomials $\{p_{l-1}(x)\}$ are appropriately chosen. This follows from the fact that by choosing \cite[Remark 7]{FI16}
\begin{equation}\label{3.7}
p_{2n}(z)=z^{2n},\quad p_{2n+1}(z)=z^{2n+1}-\left\langle\mathrm{Tr}\,P_m^2\right\rangle_{2n\times2n}z^{2n-1}
\end{equation}
the matrix $[\zeta^2 \alpha_{j,l} + \beta_{j,l} ]$ in \eqref{3.5} and \eqref{3.6} in that case $\zeta=1$ becomes block diagonal, with blocks
\begin{equation}\label{3.8}
\begin{bmatrix}0&h_{j-1}\\-h_{j-1}&0\end{bmatrix},\quad h_{j-1}=\alpha_{2j-1,2j}+\beta_{2j-1,2j},
\end{equation}
($j=1,\ldots,[N/2]$) and the last diagonal entry $0$ for $N$ odd. The choice \eqref{3.7} specifies $\{p_{l-1}(x)\}$ as skew orthogonal polynomials. This structure is the key in progressing from the Pfaffian expressions to the computation of the probabilities $p_{N,k}^{P_m}$, and as we will see later, the correlation functions.

\begin{lemma}
Consider the product \eqref{1.9}. Let each $X_i$ be a $2n\times2n$ sub-block of an $(L_i+2n)\times(L_i+2n)$ real orthogonal matrix chosen with Haar measure. We have
\begin{equation}\label{3.9}
\left\langle\mathrm{Tr}\,P_m^2\right\rangle_{2n\times2n}=\prod_{i=1}^m\frac{2n}{L_i+2n}
\end{equation}
and consequently the skew orthogonal polynomials \eqref{3.7} read
\begin{equation}\label{3.10}
p_{2n}(z)=z^{2n},\quad p_{2n+1}(z)=z^{2n+1}-\left(\prod_{i=1}^m\frac{2n}{L_i+2n}\right)z^{2n-1}.
\end{equation}
Furthermore, the normalisation in \eqref{3.8} is given by
\begin{equation}\label{3.11}
h_l=\prod_{i=1}^m\frac{L_i!(2l)!}{(L_i+2l)!}.
\end{equation}
\end{lemma}

\begin{proof}
Averaging over individual elements $(X_i)_{jk}$ of each $X_i$ gives zero, while
\begin{equation}\label{3.12}
\left\langle\left((X_i)_{jk}\right)^2\right\rangle=\frac{1}{L_i+2n}.
\end{equation}
This latter fact follows from each element of $X_i$ being an element of a $(L_i+2n)\times(L_i+2n)$ Haar distributed real orthogonal matrix, and knowledge of the distribution of the moments of the latter; see e.g. \cite[Eq. (5.2) with $p=1$]{DF17}. The only time no individual terms $(X_i)_{jk}$ appear in $P_m^2$ is on the diagonal, so we have
\begin{equation*}
\left\langle\mathrm{Tr}\,P_m^2\right\rangle=\sum_{l=1}^{2n}\left\langle\left((P_m)_{ll}\right)^2\right\rangle_{X_1,\ldots,X_m}.
\end{equation*}
Now $(P_m)_{ll}$ consists of a sum of a total of $(2n)^{m-1}$ terms, each of which is a product of $m$ elements, one from each of $X_1,\ldots,X_m$. Only the square of each of these terms is non-zero in the averaging. Using \eqref{3.12} shows
\begin{equation*}
\left\langle\mathrm{Tr}\,P_m^2\right\rangle=(2n)^{m-1}\sum_{l=1}^{2n}\prod_{i=1}^m\frac{1}{L_i+2n},
\end{equation*}
which is \eqref{3.9}. Substituting in \eqref{3.7} gives \eqref{3.10}.

In relation to the normalisation, from the definition \eqref{3.2} we must have $Z_N(1)=1$. We know that use of the skew orthogonal polynomials reduces $[\alpha_{j,l}+\beta_{j,l}]_{j,l=1,\ldots,N}$ to block diagonal form with blocks \eqref{3.8}. Use of \eqref{3.5} then gives that for $N$ even
\begin{equation*}
1=\left(\prod_{i=1}^mK_{N,L_i}\right)\prod_{l=1}^{N/2}h_{l-1}.
\end{equation*}
Recalling \eqref{1.12}, this implies \eqref{3.11}.
\end{proof}

\subsection{Determinant formula}\label{s3.3}
The skew orthogonal polynomials \eqref{3.10} are even and odd when their degrees are even and odd respectively. We can check from \eqref{3.3} that this implies $\alpha_{j,l}+\beta_{j,l}=0$ unless the parity of $j$ and $l$ is opposite. The elements in the Pfaffian are thus vanishing in a chequerboard pattern, which allows for a reduction to a determinantal formula of a matrix with half the size as familiar from earlier studies \cite{FN07, FN08p, FI16}. For $N$ even, we have
\begin{equation}\label{3.13}
Z_N(\zeta)=\left(\prod_{i=1}^m K_{N,L_i}\right)\mathrm{det}\left[\zeta^2\alpha_{2j-1,2l}+\beta_{2j-1,2l}\right]_{j,l=1,\ldots,N/2}
\end{equation}
and for $N$ odd
\begin{equation}\label{3.14}
Z_N(\zeta)=\zeta\left(\prod_{i=1}^mK_{N,L_i}\right)\mathrm{det}\left[\left[\zeta^2\alpha_{2j-1,2l}+\beta_{2j-1,2l}\right]\left[\mu_{2j-1}\right]\right]_{\substack{j=1,\ldots,(N+1)/2,\\ l=1,\ldots,(N-1)/2}}.
\end{equation}

Let us denote by  $a_{j,k}$ the corresponding integral in~\eqref{3.3} with $p_l(x)=x^l$. Then with $G_{2m+1,2m+1}^{m+1,m}$ denoting a particular Meijer G-function (see e.g. \cite{Lu69}) we know from \cite[Eq. (2.14)]{FK16} that
\begin{multline}\label{3.15}
a_{2j-1,2k}=\alpha_{2j-1,2k}\Big\vert_{p_l(x)=x^l}
=\bigg(\prod_{\ell=1}^m\frac{L_\ell\Gamma(\frac{L_\ell}{2})\Gamma(\frac{L_\ell+1}{2})}{2\sqrt{\pi}}\bigg)\\
\times \MeijerG[\bigg]{m+1}{m}{2m+1}{2m+1}{\frac32-j,\ldots,\frac32-j;\frac{L_1}2+k,\ldots,\frac{L_m}2+k,1}%
{0,k,\ldots,k;\,\frac{3-L_1}2-j,\ldots,\frac{3-L_m}2-j}{1}.
\end{multline}
It follows that $\alpha_{2j-1,2l}$ in \eqref{3.13} and \eqref{3.14} evaluated using the skew orthogonal polynomials~\eqref{3.10} is given as a simple linear combination of $a_{2j-1,2k}$ and $a_{2j-1,2k-2}$, and thus is known explicitly in terms of Meijer G-functions. Moreover, use of the formula for $h_{j-1}$ in \eqref{3.8} together with \eqref{3.11} and the skew orthogonality 
\begin{equation*}
\alpha_{2j-1,2j}+\beta_{2j-1,2l}=0\quad(j\neq l)
\end{equation*}
allows $\beta_{2j-1,2l}$ to be eliminated, and we know too from \cite[Eq. (2.15)]{FK16} that with $p_{2j-2}(x)=x^{2j-2}$,
\begin{equation}\label{3.16}
\mu_{2j-1}=\prod_{\ell=1}^m\bigg(\frac{L_\ell\Gamma(\frac{L_\ell}{2})\Gamma(\frac{L_\ell+1}{2})}{2\sqrt{\pi}}\bigg)^{1/2}
\frac{\Gamma(j-\frac12)}{\Gamma(\frac{L_\ell}2+j-\frac12)}.
\end{equation}
As a consequence all entries in the determinant formulas \eqref{3.13} and \eqref{3.14} can be made explicit.

\begin{theorem}\label{theorem2}
Let $K_{N,L_i}$ be given by \eqref{1.12}, $a_{2j-1,2k}$ by \eqref{3.15}, $\mu_{2j-1}$ by \eqref{3.16} and $h_{j-1}$ by \eqref{3.11}. Setting
\begin{equation}
b_{j,k}(\zeta):=(\zeta^2-1)\left(a_{2j-1,2k}-\left(\prod_{i=1}^m\frac{2k-2}{L_i+2k-2}\right)a_{2j-1,2k-2}\right)+h_{j-1}\delta_{j,k},
\end{equation}
where $a_{\bullet,-2}=0$, we have for $N$ even that
\begin{equation}\label{gen-func-even}
Z_N(\zeta)=\left(\prod_{i=1}^mK_{N,L_i}\right)\mathrm{det}[b_{j,k}(\zeta)]_{j,k=1,\ldots,N/2}
\end{equation}
while for $N$ odd
\begin{equation}\label{gen-func-odd}
Z_N(\zeta)=\left(\prod_{i=1}^mK_{N,L_i}\right)\mathrm{det}\left[[b_{j,k}(\zeta)]_{\substack{j=1,\ldots,(N+1)/2\\ k=1,\ldots,(N-1)/2}} \quad [\mu_{2j-1}]_{j=1,\ldots,(N+1)/2}\right].
\end{equation}
\end{theorem}

The importance of the explicit formulae for the generating functions~\eqref{gen-func-even} and~\eqref{gen-func-odd} provided by Theorem~\ref{theorem2} is evident from~\eqref{3.2}; we can find the probability of finding exactly $k$ real eigenvalues by expanding the generating function~\eqref{gen-func-even} if $N$ even and~\eqref{gen-func-odd} otherwise. This approach is remarkably general as it is valid for any number of matrices $m\geq1$, any matrix dimension $N\geq1$, and any truncations $L_1,\ldots,L_m\geq0$. Numerical computations of these probabilities (using mathematical software such as \textsc{Mathematica} or \textsc{Maple}) is relatively fast for moderate $m$. Nonetheless, it is interesting to look for evaluations of the Meijer $G$-function in~\eqref{3.15} in terms of more elementary functions. Let us first consider the simplest case, that is the Meijer $G$-function in~\eqref{3.15} with $m=1$. Computer algebra yields
\begin{multline}\label{eval-meijerG}
\MeijerG[\bigg]{2}{1}{3}{3}{\frac32-j;\frac{L_1}2+k,1}{0,k;\,\frac{3-L_1}2-j}{1}
=\frac{\Gamma(k)}{\Gamma(k+\frac{L_1}{2})}\frac{\Gamma(j-\frac12)}{\Gamma(j-\frac12+\frac{L_1}{2})}\\
-\frac1{k\,\Gamma(\frac{L_1}2)}\frac{\Gamma(j+k-\frac12)}{\Gamma(j+k-\frac12+\frac{L_1}{2})}
\hypergeometric{3}{2}{k,j+k-\frac12,1-\frac{L_1}2}{k+1,j+k-\frac12+\frac{L_1}2}{1},
\end{multline}
where
\begin{equation}\label{hyper}
\hypergeometric{3}{2}{a_1,a_2,a_3}{b_1,b_2}{x}:=
\sum_{\ell=0}^\infty\frac{(a_1)_\ell(a_2)_\ell(a_3)_\ell}{(b_1)_\ell(b_2)_\ell}\frac{x^\ell}{\ell!}
\end{equation}
is a hypergeometric sum. The important observation is that $1-\frac{L_1}2$ appears as an upper-index in the hypergeometric function in~\eqref{eval-meijerG}. Thus, if $L_1$ is a positive even integer then the upper-index $1-\frac{L_1}2$ is a negative integer which implies that that hypergeometric sum~\eqref{hyper} terminates. Consequently, the Meijer $G$-function~\eqref{eval-meijerG} becomes a finite sum over ratios of gamma functions. More precisely, we have~\cite[eq.~(3.2)]{FK16}
\begin{equation}\label{gamma-sum}
 \MeijerG[\bigg]{2}{1}{3}{3}{\frac32-j;\frac{L_1}2+k,1}{0,k;\,\frac{3-L_1}2-j}{1}=
 \frac{\Gamma(j-\frac12)}{\Gamma(\frac{L_1}2)\Gamma(L_1+j+k-\frac32)}\sum_{\ell=1}^{L_1/2}
 \frac{\Gamma(j+k+\ell-\frac32)\Gamma(L_1-\ell)}{\Gamma(j+\ell-\frac12)\Gamma(\frac{L_1}2-\ell+1)}
\end{equation}
for $L_1$ even. Furthermore, we see that the right-hand in~\eqref{gamma-sum} is a rational number for any $j,k,\frac{L_1}2\in\Z_+$. Using this result in Theorem~\ref{theorem2} leads to the conclusion that all probabilities $p_{N,k}^{P_1}$ are rational numbers as long as $L_1$ is an even integer. In fact, it turns out that this property is even more general: the probabilities $p_{N,k}^{P_m}$ for any $m$ is a rational number as long as $L_1,\ldots,L_m$ are even integers. This can be seen using the method presented in~\cite[\textsection 3]{FK16} inspired by a related technique used for Gaussian matrices~\cite{Ku15}. The main idea behind this method is to use the general three-term recurrence relation for Meijer $G$-functions~\cite{Lu69}
\begin{equation}
\MeijerG{m}{n}{p}{q}{a_1,\ldots,a_p}{b_1,\ldots,b_q}{z}=
\frac{\MeijerG{m}{n}{p}{q}{a_1,\ldots,a_{p-1},a_p-1}{b_1,\ldots,b_q}{z}+
\MeijerG{m}{n}{p}{q}{a_1,\ldots,a_p}{b_1,\ldots,b_{q-1},b_q+1}{z}}{a_p-b_q-1}
\end{equation}
for $n<p$ and $m<q$; together with evaluations~\cite{Lu69,FK16}
\begin{align}
\MeijerG{m+1}{m}{2m+1}{2m+1}{\frac32-j,\ldots,\frac32-j;\ell_1+k,\ldots,\ell_m+k,1}{0,k,\ldots,k;\frac32-j,\ldots,\frac32-j}{1}&=0,\\
\MeijerG{m+1}{m}{2m+1}{2m+1}{\frac32-j,\ldots,\frac32-j;k,\ldots,k,1}{0,k,\ldots,k;\frac32-j-\ell_1,\ldots,\frac32-j-\ell_m}{1}
&=\prod_{i=1}^m\frac{\Gamma(j-\frac12)}{\Gamma(j-\frac12+\ell_i)}
\end{align}
for non-negative integers $\ell_1,\ldots,\ell_m$, when not all of them are $0$. These three formulae allow us to construct a systematic reduction scheme for the Meijer $G$-functions which appear in~\eqref{3.15}; we refer to~\cite[\textsection 3]{FK16} for the details of this reduction scheme. As an example, for $m=2$ with $L_1,L_2$ being positive even integers, the Meijer $G$-function in~\eqref{3.15} reads
\begin{multline}\label{gamma-sum-2}
 \MeijerG[\bigg]{3}{2}{5}{5}{\frac32-j,\frac32-j;\frac{L_1}2+k,\frac{L_2}2+k,1}{0,k;\,\frac{3-L_1}2-j,\frac{3-L_2}2-j}{1}=
 \sum_{p=1}^{L_1/2}\sum_{q=1}^{L_2/2}
 \frac{\Gamma(L_1-p)\Gamma(j+k+p-\frac32)}{\Gamma(\frac{L_1}2)\Gamma(\frac{L_1}2-p+1)\Gamma(L_1+j+k-\frac32)}\\
 \times\frac{\Gamma(L_2-q)\Gamma(j+k+q-\frac32)}{\Gamma(\frac{L_2}2)\Gamma(\frac{L_2}2-q+1)\Gamma(L_2+j+k-\frac32)}
 \bigg(\mathcal K_{j,k}^{p,q}+\mathcal K_{j,k}^{q,p}+\frac{\Gamma(j-\frac12)^2}{\Gamma(p+j-\frac12)\Gamma(q+j-\frac12)}\bigg)
\end{multline}
with
\begin{equation}\label{K}
\mathcal K_{j,k}^{p,q}=
 \frac{\Gamma(j-\frac12)}{\Gamma(p)\Gamma(p+q+j+k-\frac32)}\sum_{\ell=1}^{q}
 \frac{\Gamma(j+k+\ell-\frac32)\Gamma(p+q-\ell)}{\Gamma(j+\ell-\frac12)\Gamma(q-\ell+1)}.
\end{equation}
We note that $\mathcal K_{j,k}^{L_1/2,L_1/2}$ is equal to the right-hand side in~\eqref{gamma-sum}; this is part of the general structure of the reduction scheme in which evaluation of the Meijer $G$-function for a given value of $m$ will include expressions for lower values of $m$. Table~\ref{table1} shows explicitly some probabilities $p_{N,k}^{P_m}$ for finding $k$ real eigenvalues.

\begin{table}[htbp]
\caption{Probabilities $p_{N,k}^{P_m}$ for $N=2,3,4$, $m=1,2,3$ and $L_1=L_2=L_3=4$.}\label{table1}
\begin{tabular}{c|c@{\qquad}c@{\qquad}c}
& $m=1$ & $m=2$ & $m=3$  \\
\hline
$p_{2,0}^{P_m}$ & $\frac{11}{35}\approx0.3143$ & $\frac{30\,641}{128\,625}\approx0.2382$ 
& $\frac{29\,654\,713}{157\,565\,625}\approx0.1882$ \\[.2em]
$p_{2,2}^{P_m}$ & $\frac{24}{35}\approx0.6857$ & $\frac{97\,984}{128\,625}\approx0.7618$ 
& $\frac{127\,910\,912}{157\,565\,625}\approx0.8118$ \\[1em]

$p_{3,1}^{P_m}$ & $\frac{73}{105}\approx0.6952$ & $\frac{1\,0968\,107}{3\,472\,875}\approx0.5667$ 
& $\frac{18\,344\,527\,259}{38\,288\,466\,875}\approx0.4791$ \\[.2em]
$p_{3,3}^{P_m}$ & $\frac{32}{105}\approx0.3048$ & $\frac{1\,504\,768}{3\,472\,875}\approx0.4333$ 
& $\frac{19\,943\,919\,616}{38\,288\,466\,875}\approx0.5209$ \\[1em]

$p_{4,0}^{P_m}$ & $\frac{421}{2205}\approx0.1909$ & $\frac{24\,149\,151\,605\,489}{214\,040\,075\,720\,625}\approx0.1128$ 
& $\frac{1\,431\,169\,011\,017\,974\,588\,501}{19\,078\,916\,984\,518\,815\,703\,125}\approx0.0750$ \\[.2em]
$p_{4,2}^{P_m}$ & $\frac{17\,576}{24\,255}\approx0.7246$ & $\frac{152\,493\,653\,488\,832}{214\,040\,075\,720\,625}\approx0.7125$  
& $\frac{140\,868\,762\,431\,563\,179\,004\,928}{209\,868\,086\,829\,706\,972\,734\,375}\approx0.6712$ \\[.2em]
$p_{4,4}^{P_m}$ & $\frac{2048}{24\,255}\approx0.0844$ & $\frac{37\,379\,270\,626\,304}{214\,040\,075\,720\,625}\approx0.1747$ 
& $\frac{53\,256\,465\,276\,946\,073\,255\,936}{209\,868\,086\,829\,706\,972\,734\,375}\approx0.2538$ 
\end{tabular}
\end{table}

Above we have seen that the probability of finding $k$ real eigenvalues is a rational number, $p_{N,k}^{P_m}\in\mathbb Q$, if all $L_1,\ldots,L_m$ are even integers. Thus, it is natural to ask if a similar phenomenon is present if one (or more) of the truncations $L_1,\ldots,L_m$ is an odd integer. The answer to this question appears to be negative. It follows from~~\cite[eq.~(3.15)]{FK16} that for $m=1$ and $L_1$ odd we have
\begin{multline}\label{gamma-sum-odd}
 \MeijerG[\bigg]{2}{1}{3}{3}{\frac32-j;\frac{L_1}2+k,1}{0,k;\,\frac{3-L_1}2-j}{1}=
\sum_{p=1}^{\frac{L_1-1}{2}}\sum_{q=1}^k\frac{\Gamma(k)\Gamma(p+q)\Gamma(L_1-p-1)\Gamma(j+k-q-\frac{1}{2})}{\Gamma(\frac{L_1-1}{2})\Gamma(\frac{L_1+1}{2}-p)\Gamma(k-q+1)\Gamma(j+k+L_1-\frac{3}{2})}\\
\times\bigg(\frac{1}{\Gamma(\frac{1}{2})\Gamma(p+q+\frac{1}{2})}+\frac{1}{\Gamma(p+\frac{1}{2})\Gamma(q+\frac{1}{2})} \bigg).
\end{multline}
Consequently, the probability of $k$ real eigenvalues will always be a polynomial in $\pi^{-1}$ with rational coefficients, e.g. for $L_1=5$ we have
\begin{equation}
p_{4,0}^{P_1}=1-\frac{385\,024}{135\,135}\frac1\pi+\frac{16\,777\,216}{18\,729\,711}\frac1{\pi^2}.
\end{equation}
We have been unable to find a systematic reduction scheme for larger $m$ when one (or more) of the truncations $L_1,\ldots,L_m$ is an odd integer. However, the result for probability of all eigenvalues real, as derived in~\cite[\textsection 3]{FK16} for $m=2, L_1=1,L_2=2$, indicates that the structure becomes more involved for larger $m$. For example, $p_{2,0}^{P_2}=1-(2\mathcal{G}+5)/(3\pi)$ and $p_{2,2}^{P_2}=(2\mathcal{G}+5)/(3\pi)$, where $\mathcal{G}\approx 0.915966$ is the Catalan's constant.

\section{Eigenvalue density}\label{s4}

\subsection{Pfaffian structure}\label{s4.1}
The Pfaffian formulae of Proposition \ref{prop1} are indicative of the fact, first identified in \cite{IK14}, that the eigenvalues of a product of truncated real orthogonal matrices form a Pfaffian point process. This means that the $k$-point correlations between real-real, real-complex and complex-complex eigenvalues are determined by correlation kernels depending only on two variables, and the number of eigenvalues, but not $k$. For example, focusing attention on the real eigenvalues, one has
\begin{equation}\label{4.1}
\rho_{(k)}^{\mathrm{real}}(x_1,\ldots,x_k)=\mathrm{Pf}\left[K^{rr}(x_j,x_l)\right]_{j,l=1,\ldots,k}
\end{equation}
with correlation kernel
\begin{equation}\label{4.2}
K^{rr}(x,y)=\begin{bmatrix}D(x,y)&S(x,y)\\-S(y,x)&\tilde{I}(x,y)\end{bmatrix}.
\end{equation}
Here $D(x,y)$ and $\tilde{I}(x,y)$ are antisymmetric functions of $x$ and $y$.

Significantly, the quantities in \eqref{4.2} are known explicitly in terms of skew orthogonal polynomials. We will focus on $S(x,y)$, which according to \eqref{4.1} and the anti-symmetry of $D(x,y)$ and $\tilde{I}(x,y)$, determines the density of real eigenvalues according to
\begin{equation}\label{4.3}
\rho_{(k)}^{r}(x)=S(x,x).
\end{equation}

\begin{proposition}
We have
\begin{equation}\label{4.8}
S(x,y)=\int_{-1}^{1}\,(x-v)\mathrm{sgn}(y-v)w_r^{(m)}(x)w_r^{(m)}(v)
\sum_{j=0}^{N-2}\bigg(\prod_{i=1}^m\frac{(L_i+j)!}{L_i!j!}\bigg)(xv)^j\mathrm{d}v,
\end{equation}
(cf.~the corresponding result for Gaussian matrices~\cite[Eq. (4.9)]{FI16}).
\end{proposition}

\begin{proof}
Let $w_r^{(m)}(x)$ be given by \eqref{2.3}. With $\{p_j(x)\}$ given by \eqref{3.10}, define
\begin{equation}
q_j(x)=w_r^{(m)}(x)p_j(x)
\qquad\text{and}\qquad
\tau_j(x)=-\half\int_{-1}^1\,\mathrm{sgn}(x-u)q_j(u)\mathrm{d}u.\label{4.4}
\end{equation}
In this notation, and introducing too $h_l$ as given by \eqref{3.11}, we have that for $N$ even (see e.g. \cite[\S4.5]{Ma11})
\begin{equation}\label{4.5}
S(x,y)=2\sum_{j=0}^{N/2-1}\frac{1}{h_j}\left(q_{2j}(x)\tau_{2j+1}(y)-q_{2j+1}(y)\tau_{2j}(x)\right).
\end{equation}
In the case of $N$ odd, let $\mu_{2j-1}$ be given by \eqref{3.16} and set $\mu_{2j}=0$. Then with
\begin{equation}
\hat{q}_j(x)=q_j(x)-\frac{\mu_{j+1}}{\mu_N}q_{N-1}(x)
\qquad\text{and}\qquad
\hat{\tau}_j(x)=-\half\int_{-1}^1\,\mathrm{sgn}(x-u)\hat{q}_j(u)\mathrm{d}u\label{4.6}
\end{equation}
we have (see e.g. \cite[\S4.6]{Ma11})
\begin{equation}\label{4.7}
S(x,y)=2\sum_{j=0}^{(N-1)/2-1}\frac{1}{h_j}\left(\hat{q}_{2j}(x)\hat{\tau}_{2j+1}(y)-\hat{q}_{2j+1}(x)\hat{\tau}_{2j}(y)\right)+\frac{q_{N-1}(x)}{\mu_N}.
\end{equation}
Inserting the explicit form of the skew orthogonal polynomials and their normalisation gives (\ref{4.8}).
\end{proof}

We know (e.g.~from \cite[\S4.6]{Ma11}) that the $k$-point correlation for the real-complex and the complex-complex eigenvalues has the same formal structure as \eqref{4.1} and \eqref{4.2}, with the matrix elements in \eqref{4.2} again permitting a single sum expression in terms of skew orthogonal polynomials. The simplest of these is $S(w,z)$ for the complex-complex correlation, where after simplification of the general formula we find
\begin{equation}\label{4.9}
S(w,z)=2\mathrm{i}\left(w_c^{(m)}\left((u,v)\right)w_c^{(m)}\left((x,y)\right)\right)^{1/2}
\sum_{j=0}^{N-2}\bigg(\prod_{i=1}^m\frac{(L_i+j)!}{L_i!j!}\bigg)(\bar{z}-w)(w\bar{z})^j,
\end{equation}
which again should be compared to the related result for Gaussian matrices~\cite[Eq. (4.4)]{FI16}. As for the real eigenvalues, the spectral density for the complex eigenvalues is given in terms of the kernel $S$,
\begin{equation}
\rho_{(1)}^{c}(z)=S(z,z).
\end{equation}
In the next subsection, we will see how the spectral density of the real eigenvalues~\eqref{4.3} can be used to get an explicit expression for the average number of real eigenvalues.

\subsection{Average number of real eigenvalues}\label{s4.4}

Integrating (\ref{4.3}) over $x$ gives the expected number of real eigenvalues, $\mathbb{E}(\#\mathrm{reals})$ say. 

\begin{corollary}
We have
\begin{equation}\label{expect-real}
\mathbb{E}(\#\mathrm{reals})=2\sum_{j=0}^{N-2}(-1)^j
\prod_{i=1}^m\binom{L_i+j}{L_i}a_{2\lceil\frac j2+1\rceil-1,2\lfloor\frac j2+1\rfloor},
\end{equation}
where according to~\eqref{3.15} we have
\begin{multline}\label{a-real}
a_{2\lceil\frac j2+1\rceil-1,2\lfloor\frac j2+1\rfloor}=
\bigg(\prod_{\ell=1}^m\frac{L_\ell\Gamma(\frac{L_\ell}{2})\Gamma(\frac{L_\ell+1}{2})}{2\sqrt{\pi}}\bigg)\\
\times\MeijerG[\Bigg]{m+1}{m}{2m+1}{2m+1}%
{\frac{1}2,\ldots,\frac12;\frac{L_1}2+j+1,\ldots,\frac{L_m}2+j+1,\lceil\frac j2\rceil+1}%
{\lceil\frac j2\rceil,j+1,\ldots,j+1;\frac{1-L_1}2,\ldots,\frac{1-L_m}2}{1}.
\end{multline}
\end{corollary}

\begin{proof}
Setting $x=y$ in \eqref{4.8} gives an explicit formula for $S(x,x)$, and we compute
\begin{equation}\label{4.20}
\mathbb{E}(\#\mathrm{reals})=\int_{-1}^1\,\rho_{(1)}^r(x)\mathrm{d}x=2\sum_{j=0}^{N-2}\prod_{i=1}^m\binom{L_i+j}{L_i}a_{j+1,j+2}
\end{equation}
where
\begin{equation}\label{4.21}
a_{j,k}=\int_{-1}^1\,\mathrm{d}x\int_{-1}^1\,\mathrm{d}y \,w_r^{(m)}(x)w_r^{(m)}(y)x^{j-1}y^{k-1}\mathrm{sgn}(y-x).
\end{equation}
We know from \eqref{3.14} an evaluation of $a_{j,k}$ in the cases that $j$ is odd and $k$ is even; this in fact covers all cases required by \eqref{4.20} due to the anti-symmetry $a_{j,k}=-a_{k,j}$. Using the latter we rewrite
(\ref{4.20}) as  (\ref{expect-real}).

Formula~\eqref{a-real} is identical to~\eqref{3.15} with $j\mapsto \lceil\frac j2+1\rceil$ and $k\mapsto \lfloor\frac j2+1\rfloor$, where have used the identity
\begin{equation}
z^\rho\MeijerG[\bigg]{m}{n}{p}{q}{a_1,\ldots,a_p}{b_1,\ldots,b_q}{z}
=\MeijerG[\bigg]{m}{n}{p}{q}{a_1+\rho,\ldots,a_p+\rho}{b_1+\rho,\ldots,b_q+\rho}{z}
\end{equation}
in order to simplify the indices. 
\end{proof}

We remark that if $L_1,\ldots,L_m$ are even positive integers, then we can evaluate the Meijer $G$-function in~\eqref{a-real} using a similar method as in Section~\ref{s3.3} which allows us to find explicit evaluations of the average number of real eigenvalues~\eqref{expect-real} as finite sums over ratios of gamma functions. For instance, for $m=1$ with $L_1$ even, it follows from~\eqref{gamma-sum} that
\begin{equation}\label{Ereal-m=1}
 \E(\#\text{reals})=\frac{\Gamma(\frac{L_1+1}2)}{\sqrt{\pi}\Gamma(L_1)}\sum_{j=0}^{N-2}\sum_{\ell=1}^{L_1/2}
 \frac{(-1)^j\Gamma(j+L_1+1)\Gamma(\lceil\frac j2\rceil+\frac12)\Gamma(j+\ell+\frac12)\Gamma(L_1-\ell)}%
 {\Gamma(j+1)\Gamma(j+L_1+\frac12)\Gamma(\ell+\lceil\frac j2\rceil+\frac12)\Gamma(\frac{L_1}2-\ell+1)}.
\end{equation}
The result for $m=1$ with $L_1$ odd can similarly be written down using~\eqref{gamma-sum-odd}. For $m=2$ with  $L_1,L_2$ even, it follows from~\eqref{gamma-sum-2} that
\begin{multline}\label{Ereal-m=2}
\E(\#\text{reals})=\frac{\Gamma(\frac{L_1+1}2)\Gamma(\frac{L_2+1}2)}{\pi\Gamma(L_1)\Gamma(L_2)}
\sum_{j=0}^{N-2}\sum_{p=1}^{L_1/2}\sum_{q=1}^{L_2/2}(-1)^j
\frac{\Gamma(L_1-p)\Gamma(j+L_1+1)\Gamma(j+p+\frac12)}{\Gamma(j+1)\Gamma(j+L_1+\frac12)\Gamma(\frac{L_1}2-p+1)}\\
\times\frac{\Gamma(L_2-q)\Gamma(j+L_2+1)\Gamma(j+q+\frac12)}{\Gamma(j+1)\Gamma(j+L_2+\frac12)\Gamma(\frac{L_2}2-q+1)}
\bigg(\mathcal K_j^{p,q}+\mathcal K_j^{q,p}+
\frac{\Gamma(\lceil\frac j2\rceil+\frac12)^2}{\Gamma(p+\lceil\frac j2\rceil+\frac12)\Gamma(q+\lceil\frac j2\rceil+\frac12)}\bigg)
\end{multline}
with
\begin{equation}
\mathcal K_{j}^{p,q}:=\mathcal K_{\lceil\frac j2+1\rceil,\lfloor\frac j2+1\rfloor}^{p,q}=
 \frac{\Gamma(\lceil\frac j2\rceil+\frac12)}{\Gamma(p)\Gamma(p+q+j+\frac12)}\sum_{\ell=1}^{q}
 \frac{\Gamma(j+\ell+\frac12)\Gamma(p+q-\ell)}{\Gamma(\lceil\frac j2\rceil+\ell+\frac12)\Gamma(q-\ell+1)}.
\end{equation}
We observe that both~\eqref{Ereal-m=1} and~\eqref{Ereal-m=2} are rational numbers for all $N,\frac{L_1}2,\frac{L_2}2\in\Z_+$. This is no surprise since we already know that all probabilities $p_{N,0}^{P_m},\ldots,p_{N,N}^{P_m}$ are rational numbers whenever $L_1,\ldots,L_m$ are positive even integers and 
\begin{equation}
\E(\#\text{reals})=\sum_{k=1}^Nk\,p_{N,k}^{P_m}.
\end{equation}
Thus, by the same arguments as in Section~\ref{s3.3}, we know that the average number of real eigenvalues (for any $m,N\geq1$) will be given by a rational number when $L_1,\ldots,L_m$ are even integers. Table~\ref{table2} shows explicit evaluations for the average number of real eigenvalues in a few cases; it is easily verified that the averages presented in Table~\ref{table2} are consistent with the probabilities from Table~\ref{table1}.

\begin{table}[htbp]
\caption{Average number of real eigenvalues for $m=1,2,3$, $N=2,3,4$ and either $L_1=L_2=L_3=4$.}\label{table2}
 \begin{tabular}{c|@{\quad}c@{\qquad}c@{\qquad}c}
   & $m=1$          & $m=2$      &     $m=3$         \\
  \hline
  $N=2$ & $\frac{48}{35}\approx1.3714$ &  
  $\frac{195\,968}{128\,625}\approx1.5236$ &
  $\frac{255\,821\,824}{157\,565\,625}\approx1.6236$  \\[.2em]
  $N=3$ & $\frac{169}{105}\approx 1.6095$ &  
  $\frac{6\,482\,411}{3\,472\,875}\approx1.8666$ &
  $\frac{78\,176\,286\,107}{38\,288\,446\,875}\approx2.0418$  \\[.2em]
  $N=4$ & $\frac{688}{385}\approx 1.7870$ &  
  $\frac{9\,817\,004\,416}{4\,622\,396\,625}\approx2.1238$ &
  $\frac{14\,537\,252\,216\,952\,832}{6\,166\,392\,657\,665\,625}\approx2.3575$  \\[.2em]
 \end{tabular}
\end{table}

\subsection{A single truncated real orthogonal matrix}\label{s4.2}

We will now look at asymptotic properties and we will start with the $m=1$ case which is the simplest by far. This simplicity,
first identified in \cite{KSZ09}, is due to the simple form of the weights $w_r^{(1)}$ and $w_c^{(1)}$.
Thus, it follows from \eqref{1.10} and \eqref{1.13} that
\begin{equation}\label{4.10}
w_r^{(1)}(\lambda)=\frac{(1-\lambda^2)^{L/2-1}}{\sqrt{2\pi}}\left(L\frac{\Gamma(1/2)\Gamma\left((L+1)/2\right)}{\Gamma(L/2)}\right)^{1/2},\quad |\lambda|<1
\end{equation}
and
\begin{equation*}
w_c^{(1)}\left((x,y)\right)=\left\{\begin{array}{ll}\frac{L(L-1)}{2\pi}|1-z^2|^{L-2}\int_{2|\mathrm{Im}\,z|/|1-z^2|}^1(1-t^2)^{(L-3)/2}\mathrm{d}t,&\quad L>1
\\ \frac{1}{2\pi}\frac{1}{|1-z^2|},&\quad L=1,\end{array}\right.
\end{equation*}
where $z=x+\mathrm{i}y\in D_{+}$.

Simplified, summed up expressions are also known for $S(x,x)$ (real case) and $S(z,z)$ (complex case), or equivalently for the real and complex eigenvalue densities. For this introduce the incomplete beta integral
\begin{equation*}
I_s(a,b)=\frac{1}{B(a,b)}\int_0^s\,t^{a-1}(1-t)^{b-1}\mathrm{d}t
\qquad\text{and}\qquad
B(a,b)=\int_0^1\,t^{a-1}(1-t)^{b-1}\mathrm{d}t.
\end{equation*}
We know from \cite{KSZ09,Ma11,Fi12} that
\begin{equation}\label{4.11}
S(x,x)=\frac{1}{B(L/2,1/2)}\frac{I_{1-x^2}(L+1,N-1)}{1-x^2}
+\frac{(1-x^2)^{(L-2)/2}|x|^{N-1}}{B(N/2,L/2)}I_{x^2}\left(\frac{N-1}2,\frac{L+2}2\right),
\end{equation}
and from \cite[Eq. (372)]{Ma11} that
\begin{equation}\label{4.12}
S(z,z)=\frac{2\mathrm{Im}(z) L(L-1)}{\pi}\frac{|1-z^2|^{L-2}}{(1-|z|^2)^{L+1}}
\int_{\frac{2|\mathrm{Im}(z)|}{|1-z^2|}}^1(1-t^2)^{(L-3)/2}\mathrm{d}t
\left(1-I_{|z|^2}(N-1,L+1)\right),
\end{equation}
for $L>1$ and
\begin{equation}\label{4.13}
S(z,z)=\frac{2\mathrm{Im}(z)}{\pi|1-z^2|(1-|z|^2)^2}\left(1-N|z|^{2N-2}+(N-1)|z|^{2N}\right)
\end{equation}
for $L=1$.

The sought asymptotic properties can now be deduced from the above exact formulae.

\begin{proposition}[Khoruzhenko et al.~\cite{KSZ09}]\label{prop:density-alpha-fixed}
Let $\alpha=N/(N+L)$, and let $\chi_J$ denote the indicator function for the interval $J$.
For $\alpha\in(0,1)$ fixed, the asymptotic density of the real eigenvalues reads
\begin{equation}\label{real-density}
\lim_{N\to\infty}\frac1{\sqrt N}\rho_{(1)}^r(x)=\sqrt{\frac{(1-\alpha)}{\pi\alpha}}\frac{1}{1-x^2} \, \chi_{-\sqrt{\alpha}<x<\sqrt{\alpha}},
\end{equation}
Likewise, the density of the complex eigenvalues with $\alpha\in(0,1)$ fixed is
\begin{equation}\label{global1}
\lim_{N\to\infty}\frac1N\rho_{(1)}^c(z)=\frac{(1-\alpha)}{\pi\alpha}\frac{1}{(1-|z|^2)^2} \, \chi_{|z|<\sqrt\alpha}.
\end{equation}
\end{proposition}

\begin{proof}
The starting point is the exact formulae~\eqref{4.11} and~\eqref{4.12}. First, use that $L=N(1-\alpha)/\alpha$. In the asymptotic limit of the integrals in~\eqref{4.11} and~\eqref{4.12} can now be evaluated using the method of steepest descent, see~\cite[\textsection 7.6.1]{Ma11}.
\end{proof}

\begin{corollary}\label{col:Em=1}
Under the same assumptions as in Proposition~\ref{prop:density-alpha-fixed}, the average number of real eigenvalues grow asymptotically as
\begin{equation}
\E(\#\textup{reals})\sim2\sqrt{\frac{N(1-\alpha)}{\pi\alpha}}\mathrm{artanh}\sqrt{\alpha}.
\end{equation}
\end{corollary}

\begin{proof}
Follows by integration over $x$ in~\eqref{real-density}, since $\E(\#\text{reals})=\int dx \rho_{(1)}^r(x)$.
\end{proof}

The large-$N$ limit is different when $L$ (rather than $\alpha$) is kept fixed. In this limit, the real and complex densities, \eqref{4.11} and \eqref{4.12}, develop singular behaviour on the boundaries, i.e. at $x=\pm1$ and $|z|=1$ for the real and complex density, respectively. Thus, we need to look at the neighbourhood of the boundaries in order to see any non-trivial behaviour. In the real case the following result holds.

\begin{proposition}[Khoruzhenko et al.~\cite{KSZ09}]\label{prop:density-alpha-1}
For $L>1$ fixed and $x\in(0,\infty)$, we have
\begin{equation}
 \lim_{N\rightarrow\infty}\frac1N\rho_{(1)}^r\Big(1-\frac xN\Big)=\tilde\rho_{(1)}^r(x)
\end{equation}
with
\begin{equation}\label{4.17}
 \tilde\rho_{(1)}^r(x)
 =\frac{x^{L/2-1}e^{-x}}{2\Gamma(L/2)}\left(1-x^{L/2+1}\overline{\gamma}(L/2+1,x)+\frac{(2x)^L}{B(L/2,1/2)}\overline{\gamma}(L+1,2x)\right),
\end{equation}
where $\overline\gamma(n,x)=(x^n\Gamma(n))^{-1}\int_0^xt^{n-1}e^{-t}dt$.
\end{proposition}

It is seen that to leading order the tail behaviour of the density~\eqref{4.17} is
\begin{equation}\label{4.18}
\lim_{N\rightarrow\infty}\tilde{\rho}_{(1)}^r(x)\underset{x\rightarrow\infty}{\sim}\frac{1}{B(L/2,1/2)x}.
\end{equation}
It follows that $\int \tilde{\rho}_{(1)}^r(x)dx$ diverges, which tells us that average  number of real eigenvalues grows with $N$. However, unlike the case with $\alpha$ fixed, the average number number of real eigenvalues cannot be obtained by a trivial integration. Nonetheless, it is claimed in~\cite{KSZ09} that
\begin{equation}
\E(\#\text{reals})\underset{N\to\infty}{\sim}\frac{\log N}{B(L/2,1/2)}.
\end{equation}

We know from \cite{FK16} that the probability $p_{N,N}^{P_1}$ of all eigenvalues of a single truncated real orthogonal
matrix being real, which is the coefficient of $\zeta^N$ in (\ref{3.13}) and (\ref{3.14}) with $m=1$, can be written
in the product form (we write $L_1 = L$)
\begin{equation}\label{pnn1}
p_{N,N}^{P_1} = \prod_{j=0}^{N-1} \frac{\Gamma(L+j) \Gamma((L+j)/2)}{\Gamma(L+(N+j-1)/2) \Gamma(L/2)}.
\end{equation}
Introducing the Barnes $G$-function $G(z)$ by its relation to the gamma function,
$G(z+1) = \Gamma(z) G(z)$, we see that (\ref{pnn1}) can be written
\[
p_{N,N}^{P_1} = \frac{1}{(\Gamma(L/2))^N}
\frac{G((2L+N)/2)}{G(L)}
\frac{G((N+L)/2)}{G(L/2)}
\frac{G((N+L + 1)/2)}{G((L+1)/2)}
\frac{G((2L+N-1)/2)}{G(N+L-1/2)}.
\]
Using the asymptotic expansion \cite{Ba00}
\begin{equation}
\log G(x) \mathop{\sim}\limits_{x \to \infty} \frac{x^2}{2} \log x - \frac{3}{4} x^2 + {\rm O} ( x \log x)
\end{equation}
in this formula allows the leading large $N,L$ asymptotic form of $p_{N,N}^{P_1}$ to be deduced.

\begin{proposition}
Let $L = c N$ with $c>0$ fixed. We have
\begin{multline}\label{pnn2}
p_{N,N}^{P_1}\mathop{\sim}\limits_{N\to\infty}\exp\Big\{N^2\Big(-\frac{c}{4}-\frac{1}{4}\log2-\frac{c}{2}\log c-\frac{3c^2}{4}\log c\\
-\frac14(c+1)^2\log(c+1)+(c+\tfrac12)^2\log(c+\tfrac12)\Big)
+{\rm O}(N\log N)\Big\}.
\end{multline}
\end{proposition}

We remark that in the limit $c \to \infty$, in which case the entries of $P_1$ approach independent standard
Gaussians, (\ref{pnn2}) reduces to $p_{N,N}^{P_1}  \sim e^{-(N^2/4) \log 2}$. This latter result is consistent with the exact result
$p_{N,N}^{P_1}  = 2^{-N(N-1)/4}$ for $P_1$ a real standard Gaussian matrix \cite{Ed97}.

For $P_1$ a real standard Gaussian matrix, the probability $p_{N,0}^{P_1}$ (with $N$ even) that all eigenvalues are {\it complex}
has the large $N$ expansion \cite{KPTTZ15, Fo15e}
\begin{equation}\label{pnn3}
\frac{1}{\sqrt{N}}\log p_{N,0}^{P_1}=-\frac{1}{\sqrt{2\pi}}\zeta\Big(\frac32\Big)+\frac{C}{\sqrt{N}}+\cdots,
\end{equation}
where $\zeta(x)$ denotes the Riemann zeta function, and $C$ is an explicit constant with the numerical value
$0.0627\cdots$. Notice in particular the proportionality of $ \log p_{N,0}^{P_1}$ on $\sqrt{N}$ rather than $N^2$
as for $\log p_{N,N}^{P_1}$. One might speculate that this is a feature too of the large $N$ form of 
$\log p_{N,0}^{P_1}$ for $P_1$ a truncated real orthogonal matrix, although an analysis based on 
(\ref{gen-func-even}) and (\ref{gen-func-odd}) appears out of reach at the present time.

\subsection{Asymptotic behaviour of general \textit{m}}\label{s4.3}

In order to study the large-$N$ asymptotic behavior for general $m>1$, we take $L_1=\cdots=L_m=L$ in the following. 

It is generally believed that under weak assumptions the global spectral density of a product of $m$ independent and identically distributed random matrices is the same as for $m$-th power of a random matrix drawn from the same ensemble;
see e.g.~\cite{BNS12}. In the $m=1$ case, we know the global spectral density for a truncated orthogonal random matrix for $N,L\to\infty$ with $\alpha\in(0,1)$ fixed, see~\eqref{global1}. From this we compute that for the $m>1$ cases,
\begin{equation}\label{global-m}
\lim_{N\to\infty}\frac1N\rho_{(1)}^c(z)=\frac{1-\alpha}{m\pi\alpha}\frac{|z|^{2/m-2}}{(1-|z|^{2/m})^2} \, \chi_{|z|<\alpha^{m/2}}.
\end{equation}
We note that the density~\eqref{global-m} is the same as for products of truncated unitary matrices~\cite{BNS12,ABKN14}. In fact, as suggested in \cite{BNS12}, the validity of~\eqref{global-m} can be proven using techniques from free probability; figure~\ref{fig1} shows a comparison with numerical data.
\begin{figure}[htbp]
\centering
\includegraphics{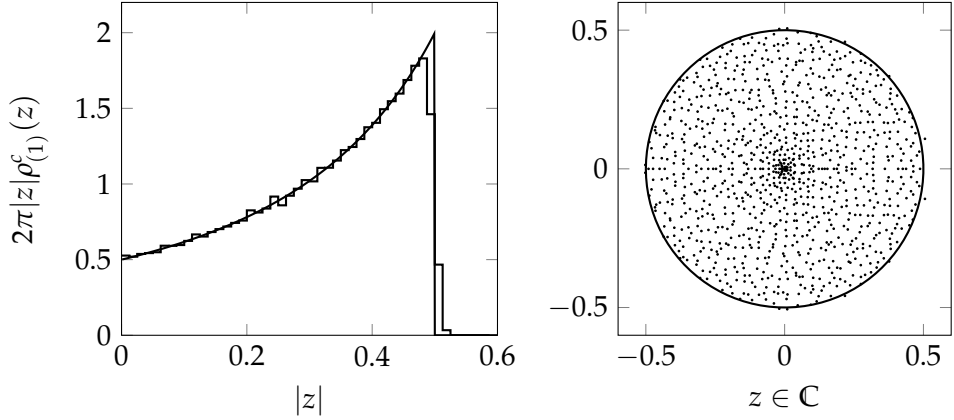}
\caption{The figure shows numerical data for a product of two (i.e. $m=2$) random truncated orthogonal matrices with $N=L=1000$. The left panel shows a histogram for absolute value of the eigenvalues (real eigenvalues are included but are sub-dominant) generated from $100$ realisation; the solid curve show the predicted large-$N$ density given by~\eqref{global-m}. The right panel shows a scatter plot of eigenvalues for a single realisation together with a circle of radius $\alpha^{m/2}$ which is the predicted radius of support as $N$ tends to infinity.}\label{fig1}
\end{figure}

When looking at the global spectrum for the real eigenvalues, we cannot employ techniques from free probability as the number of real eigenvalues are sub-dominant in $N$. This makes the real case more challenging. However, it is still believed that the global density (up to an overall normalisation) is same the $m$-th power of a single truncated orthogonal matrix, which gives us a conjecture for the density.
\begin{conjecture}\label{con1}
The normalised global spectral density for the real eigenvalues for $N,L\to\infty$ with $\alpha\in(0,1)$ and $m\geq1$ fixed is given by
\begin{equation}
\frac{\rho_{(1)}^r(x)}{\E(\#\text{reals})}=\frac{1}{2\,m\,\textrm{artanh}\sqrt\alpha}\,\frac{|x|^{1/m-1}}{(1-|x|^{2/m})}
\, \chi_{|x|<\alpha^{m/2}}.
\end{equation}
\end{conjecture}
We note that in the small-$\alpha$ limit (i.e. $L\gg N$), we have
\begin{equation}
\lim_{\alpha\to0}\frac{\alpha^{m/2}\rho_{(1)}^r(\alpha^{m/2}x)}{\E(\#\text{reals})}=\frac{|x|^{1/m-1}}{2\,m} \,
\chi_{-1<x< 1},
\end{equation}
which we recognise as the global spectral density for the real eigenvalues of a product Gaussian matrices (conjectured in~\cite{FI16} and proven by Simm in~\cite{Si16}). This is consistent with a known transition from truncated orthogonal matrices to real Gaussian matrices for $L\gg N$. Figure~\ref{fig3} verifies that there is good agreement between numerical data and Conjecture~\ref{con1}. 
\begin{figure}[htbp]
\centering
\includegraphics{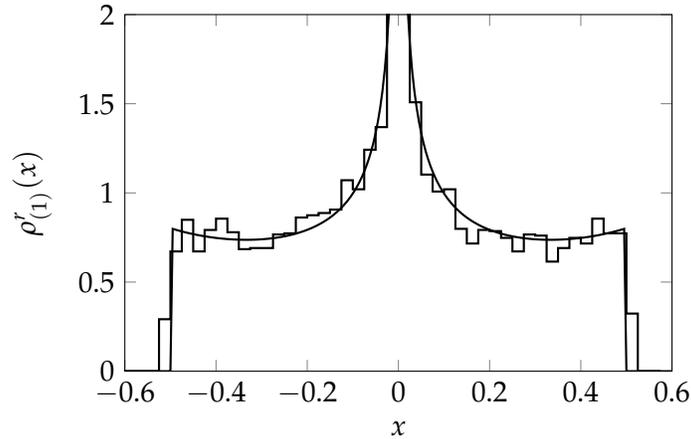}
\caption{The figure shows a histogram of the real eigenvalues for $100$ realisations of a product of two (i.e. $m=2$) random truncated orthogonal matrices with $N=L=1000$. The solid curve show the conjectural large-$N$ density given by Conjecture~\ref{con1}.}\label{fig3}
\end{figure}


Based on known behaviour for $m=1$ (Corollary~\ref{col:Em=1}) as well as known results in the Gaussian case~\cite{Si16}, we furthermore state the following conjecture.
\begin{conjecture}\label{con2}
For $N,L\to\infty$ with $\alpha=N/(N+L)\in(0,1)$ and $m>0$ fixed, the average number of real eigenvalues grows asymptotically as 
\begin{equation}
\E(\#\text{reals})\sim 2\sqrt{\frac{mN(1-\alpha)}{\pi\alpha}}\textrm{artanh}\sqrt\alpha.
\end{equation}
\end{conjecture}
Figure~\ref{fig2} shows that there is agreement between Conjecture~\ref{con2} and numerical data; we recall that asymptotic behaviour for the $m=1$ case (also shown on figure~\ref{fig2}) is known to be true (i.e.~this case is not conjectural). 
\begin{figure}[htbp]
\centering
\includegraphics{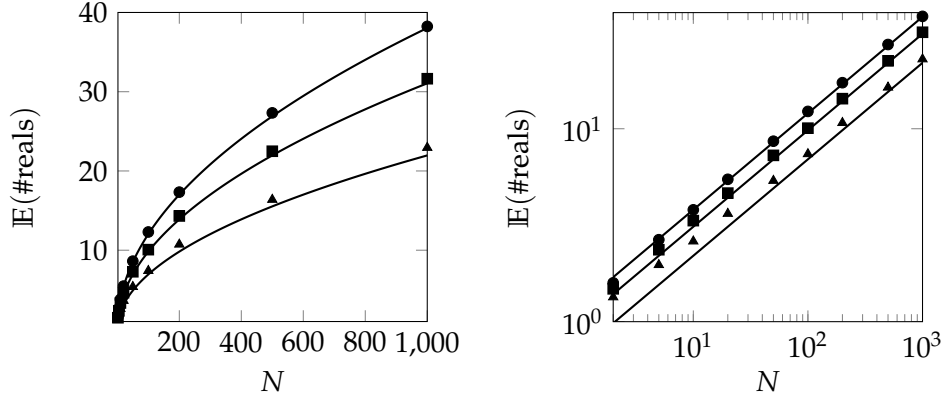}
\caption{The figure compare numerical data for the average number of real eigenvalues (data points) with the conjectural asymptotic behaviour from Conjecture~\ref{con2} (solid curves). Numerical data is provided for $m=1,2,3$ (denoted by $\blacktriangle,\blacksquare,\bullet$, respectively) and $N=L=2,5,10,20,50,100,200,500,1000$ (using $50\,000,20\,000,10\,000,5000,2000,1000,500,200,100$ realisations, respectively). The data depicted on the two panels are the same but the left panel has axes with linear scale while the right panel is double-logarithmic.}
\label{fig2}
\end{figure}

It would interesting to see if the method presented in~\cite{Si16} for Gaussian matrices could be extended to prove Conjecture~\ref{con1} and~\ref{con2} for truncated orthogonal matrices. However, such an analysis is beyond the scope of the present paper and will be postponed to future work. 

An even more challenging task is to go beyond the average number of real eigenvalues (Conjecture~\ref{con2}) and ask for the probability distribution of the number of real eigenvalues as the matrix dimension tends to infinity. Under general (but not fully understood) conditions, it is believed that the number of real eigenvalues satisfy a `central limit theorem' for large matrix dimensions~\cite{MPT16}. Based on numerical evidence, we conjecture that a similar result holds for the product ensembles considered in this paper.

\begin{conjecture}\label{con3}
Let $\mathcal E$ be a random variable given by the number of real eigenvalues
of a product of truncated orthogonal matrices with parameters $N,L,\alpha,m$ defined as above.
For $N,L\to\infty$ with $\alpha\in(0,1)$ and $m\geq1$ fixed, we have
\begin{equation}\label{4.35}
\frac{\mathcal E-\E[\mathcal E]}{\sqrt{(2-\sqrt{2})\E[\mathcal E]}}\stackrel{d}\longrightarrow\mathcal N(0,1),
\end{equation}
i.e. $\mathcal E$ converges (in distribution) to a normal random variable.
\end{conjecture}

Figure~\ref{fig4} compares the Gaussian prediction from Conjecture~\ref{con3} with numerical data generated from $1000$-by-$1000$ matrices. We see that there is excellent agreement between the numerical data and the conjecture. Note that Theorem~\ref{theorem2} together with the expansion~\eqref{3.2} gives us an explicit way to determine the probability distribution for the number of real eigenvalues for finite $N$ and $L$. However, it is a highly non-trivial task to use these explicit formulae to proof the asymptotic result given by Conjecture~\ref{con3}.

We see from (\ref{4.35}) that for large $N$, ${\rm Var} \, (\mathcal E) = (2 - \sqrt{2}) \E[\mathcal E]$
independent of $m$. 
 In fact, this same proportionality has been observed in several other asymmetric random matrix ensembles both analytically~\cite{FN07,FM11,Si17,Ko15} and numerically~\cite{MPT16}. Thus, in this setting $2 - \sqrt{2}$ is believed to be a universal constant.

\begin{figure}[htbp]
\centering
\includegraphics{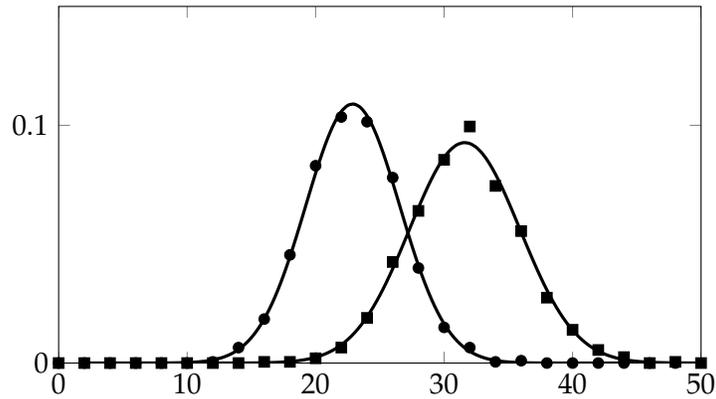}
\caption{The figure compares numerical data for the number of real eigenvalues (data points) with the conjectural asymptotic behaviour from Conjecture~\ref{con3} (solid curves). Numerical data is provided for $m=1,2$ (denoted by $\bullet$,{\tiny $\blacksquare$}, respectively) and they are generated from $1\,000$ realizations of matrices with $N=L=1000$.} 
\label{fig4}
\end{figure}

\newpage

\providecommand{\bysame}{\leavevmode\hbox to3em{\hrulefill}\thinspace}
\providecommand{\MR}{\relax\ifhmode\unskip\space\fi MR }
\providecommand{\MRhref}[2]{%
  \href{http://www.ams.org/mathscinet-getitem?mr=#1}{#2}
}
\providecommand{\href}[2]{#2}

\end{document}